\documentclass[letterpaper,11pt]{article}

\usepackage{amsthm, amssymb, amsmath, amsfonts, url, enumerate, mathrsfs}
\usepackage{fancyhdr}
\usepackage{graphics}
\usepackage{hyperref}
\usepackage{cleveref}

\usepackage[pdftex]{graphicx}
\usepackage[dvipsnames, x11names]{xcolor}

\newcounter{prn}
\newtheorem{thm}{Theorem}[section]

\newtheorem{lemma}[thm]{Lemma}
\newtheorem{defn}[thm]{Definition}
\newtheorem{corollary}[thm]{Corollary}

\newtheorem{note}[thm]{Note}

\crefname{lemma}{lemma}{lemmas}
\Crefname{lemma}{Lemma}{Lemmas}
\crefname{thm}{theorem}{theorems}
\Crefname{thm}{Theorem}{Theorems}
\crefname{corollary}{corollary}{corollaries}
\Crefname{corollary}{Corollary}{Corollaries}
\crefname{proposition}{proposition}{proposition}
\Crefname{proposition}{Proposition}{Proposition}

\def\F{{\mathbb{F}}}

\def\N{{\mathbb{N}}}

\def\C{{\mathbb{C}}}

\def\E{{\mathbb E}}

\newcommand{\st}[1]{\{#1\}}

\def\to{\rightarrow}

\makeatletter
\newcommand{\newparallel}{\mathrel{\mathpalette\new@parallel\relax}}
\newcommand{\new@parallel}[2]{%
  \begingroup
  \sbox\z@{$#1T$}
  \resizebox{!}{\ht\z@}{\raisebox{\depth}{$\m@th#1/\mkern-5mu/$}}%
  \endgroup
}
\makeatother

\newcommand{\veps}{\varepsilon}

\usepackage{amsmath,amsthm,amsfonts,amssymb,color,latexsym, mathrsfs, url, enumerate}
\usepackage{setspace,fullpage} 
\usepackage{tikz}
\usepackage{bbm}
\usetikzlibrary{quantikz2}
\usetikzlibrary{matrix, positioning}

\usepackage[margin=1in]{geometry}
\usepackage{accents}
\usepackage{fancyhdr}
\usepackage{graphicx} 

\usepackage{cleveref}
\crefname{lemma}{lemma}{lemmas}
\Crefname{lemma}{Lemma}{Lemmas}
\crefname{thm}{theorem}{theorems}
\Crefname{thm}{Theorem}{Theorems}

\graphicspath{{./assets/}}
\makeatletter
\let\@fnsymbol\@arabic

\title{Quantum Worst-Case to Average-Case Reduction for\\Matrix-Vector Multiplication
\\[0.5em]}

\author{
Divesh Aggarwal\thanks{National University of Singapore. Email: divesh.aggarwal@gmail.com}\and
Dexter Kwan\thanks{National University of Singapore. Email: dexter.kwxn@gmail.com}
}
\date{}

\begin{document}

\maketitle
\begin{abstract}
Worst-case to average-case reductions are a cornerstone of complexity theory, providing a bridge between worst-case hardness and average-case computational difficulty. While recent works have demonstrated such reductions for fundamental problems using deep tools from additive combinatorics, these approaches often suffer from substantial complexity and suboptimal overheads. In this work, we focus on the quantum setting, and provide a new reduction for the Matrix-Vector Multiplication problem that is more efficient, and conceptually simpler than previous constructions. By adapting hardness self-amplification techniques to the quantum domain, we obtain a quantum worst-case to average-case reduction with improved dependence on the success probability, laying the groundwork for broader applications in quantum fine-grained complexity.

\end{abstract}

\newpage

\tableofcontents
\newpage

\pagenumbering{arabic}
\setcounter{page}{1}

\section{Introduction}

\subsection{Background and Motivation}

Worst-case to average-case reductions form a fundamental theme in the theory of computation. These reductions transform algorithms that succeed on a small fraction of inputs into ones that succeed with high probability on every input, thereby converting worst-case hardness into average-case hardness. Such transformations have played a crucial role in cryptography, derandomization, fine-grained complexity, and post-quantum cryptographic foundations.

Classically, worst-case to average-case reductions have been established for a handful of problems, often relying on random self-reducibility (e.g., for the discrete logarithm and RSA problems) or using heavy tools from additive combinatorics, as in the recent work of Asadi, Golovnev, Gur, and Shinkar~\cite{Asadi_Golovnev_Gur_Shinkar_2022}. Their framework, based on probabilistic versions of the Bogolyubov-Ruzsa lemma, provided the first efficient worst-case to average-case reductions for a broad class of fine-grained problems, such as matrix multiplication, online matrix-vector multiplication (OMv), and various data structure problems. However, the reduction overhead incurred in their work was quasi-polynomial in the inverse success rate, namely $2^{O(\log^5(1/\alpha))}$, limiting practical applicability.

More recently, Hirahara and Shimizu~\cite{Hirahara_Shimizu_2023} proposed a fundamentally different and significantly simpler framework for hardness self-amplification. Instead of relying on additive combinatorics, they utilized upward one-query reductions and expansion properties of appropriately constructed bipartite graphs. Their techniques yield only a polynomial overhead in $1/\alpha$ and lead to conceptually cleaner proofs. However, these advances were established only in the classical setting.

In the quantum domain, Asadi, Golovnev, Gur, Shinkar, and Subramanian~\cite{Asadi_Golovnev_Gur_Shinkar_Subramanian_2022} extended the additive combinatorics-based framework to quantum algorithms, providing quantum worst-case to average-case reductions for all linear problems. Their work crucially leveraged the power of quantum singular value transformations and new quantum local correction procedures. Nevertheless, the quantum reductions inherited the heavy combinatorial complexity from their classical precursors.

\subsection{Our Contribution}

In this work, we introduce a new quantum worst-case to average-case reduction for the Matrix-Vector Multiplication problem that fundamentally improves upon prior approaches both in simplicity and in efficiency. Specifically, we adapt the hardness self-amplification techniques of Hirahara and Shimizu~\cite{Hirahara_Shimizu_2023} to the quantum setting, overcoming the technical barriers that previously necessitated intricate combinatorial arguments.

In this paper, we define oracle access for matrices and vectors as follows. We let quantum algorithms access a matrix $M\in\F^{n\times n}$ via a unitary oracle $U_M$ that does the following map
\[ U_M\ket{j, k, z} = \ket{j, k, z\oplus M_{jk}} \]
for all $j,k\in [n]$ and $z\in \F$, and $\oplus$ denotes addition over $\F$. Similarly for a vector $v\in\F^n$, oracle access means access to unitary $U_v$ such that
\[ U_v\ket{j, z} = \ket{j, z\oplus v_j} \]
for all $j\in [n]$ and $z\in \F$. In this paper, we use the notation $\Pi_{Mv}$ to denote an orthogonal projection on the output register of $\sf ALG$ which represents measuring the outcome $Mv \in \F^n$ in the standard basis. 

Our main result shows that given a quantum algorithm that correctly computes $Mv$ for a random matrix $M$ and vector $v$ with probability at least $\alpha$, one can construct a quantum algorithm that succeeds on \emph{every} input with high probability, while incurring a \emph{nearly quadratic} overhead in $1/\alpha$, namely $O(\alpha^{-2})$, as opposed to the previously known overheads of $O(\alpha^{-9/2})$ or worse. Specifically, we have the following result.
\begin{thm} [Informally stated; see \Cref{thm:main}]
\label{thm:main_informal}
    Let $\F = \F_p$ be a prime field, $n\in\N$, and $\alpha := \alpha(n) \in (0, 1]$. Suppose that there exists a quantum algorithm {\sf ALG} that has oracle access to a matrix $M\in \F^{n\times n}$ and a vector $v$, makes $Q(n)$ queries, and satisfies
\[
\Pr_{\substack{M, v, \\ \text{\sf ALG}}} \left[ \text{\sf ALG}^{M, v} = Mv \right] 
= \underset{\substack{M \in \mathbb{F}^{n \times n} \\ v \in \mathbb{F}^n}}{\mathbb{E}}
\left[ \left\| \Pi_{Mv} \text{\sf ALG}^{M, v} \ket{0} \right\|^2 \right] \geq \alpha \, .
\]
Then, for every constant $\delta > 0$, there exists a worst-case algorithm $\sf ALG'$ that makes $O(\alpha^{-2})$ queries to $\sf ALG$, makes $O((Q(n)+n^{3/2})\cdot \alpha^{-2}\cdot \text{\rm poly}\log(1/\alpha))$ queries to $U_M$ and $U_M^\dagger$, and succeeds on all inputs with high probability:
\[
\forall M\in\F^{n\times n}, v\in\F^n, \Pr_{\text{\sf ALG}'}\left[(\text{\sf ALG}')^{M,v} = Mv\right] = 
\left\| \Pi_{Mv} (\text{\sf ALG}')^{M,v} |0\rangle \right\|^2 \geq 1 - \delta.
\]
\end{thm}
\section{Proof Overview}
For a worst-case to average-case reduction, we begin with an average-case algorithm with success probability $\alpha$. We wish to amplify this success probability such that it succeeds with high probability for every input. Note that the average-case algorithm's success rate is taken over all inputs, and so, it may be that for some input matrices and vectors, we have a lower probability of success (e.g. $\alpha/n$), while for others a higher probability (e.g $2\alpha$). This means that we cannot simply repeat the algorithm and verify the answer until we get a correct answer, since we could have the worst inputs blowing up the number of iterations. Instead, we use the framework by \cite{Hirahara_Shimizu_2023}, which uses the reduction of the direct product theorem. Specifically, given an instance $x$, sample $k$ independent random instances $x_1, x_2,...,x_k$ of $f$. Then, sample $i\sim [k]$ and set $\bar{x} = (x_1, ..., x_{i-1}, x, x_{i+1}, ..., x_k)$. Run the average case solver $\mathcal{M}$ for the $k$-wise direct product problem $f^k$ and obtain $(y_1, ..., y_k) = \mathcal{M}(\bar{x})$. If $f(x) = y_i$, output $y_i$. Note that this requires a verifier for our problem.

In general, the reduction can be seen as a one-query randomized reduction $\mathcal{R}^\mathcal{M}(x)$, where $\mathcal{M}$ is an average-case solver with success probability $\alpha$. This reduction results in an edge-weighted bipartite graph $G=(X,Y,W)$, where $X, Y$ are distributions of the input spaces of $f$ and $f^k$ respectively. Denoting $P(x,y)$ to be the probability $\mathcal{R}^\mathcal{M}(x)$ produces $y$, we have that the edge weight of $(x,y)$ is $W(x,y):= \Pr(X=x)P(x,y)$. We refer to this $G$ as our query graph. \cite{Hirahara_Shimizu_2023} showed that this query graph is a $(\delta,c)$-sampler for density $\veps$, which has the property that for a $(1-\delta)$-fraction of $x\in X$, the algorithm $\mathcal{M}$ succeeds on at least $(1-c)\veps$-fraction of neighbours of $x$. This allows us to, given a verifier, repeatedly run $\mathcal{R}^\mathcal{M}(x)$ until we get a correct answer with high probability.

\subsection{Matrix-vector Multiplication (MvM) and Direct Product Theorem}

Given a matrix $M\in\F^{n\times n}$ and $v\in\F^n$, we first construct an algorithm that does matrix-vector multiplication for $M\in\F^{d\times d}$ and $v\in\F^d$ for a chosen $d$, where $d$ divides $n$, while amplifying the success rate in the process. Denoting $k = n / d$, we proceed by dividing the matrix in $k^2$ blocks and the vector input into $k$ blocks (see \Cref{fig:3}). To compute the original matrix-vector product, we need to simply compute $k^2$ many sub-products, and combine them appropriately. If we can compute each sub-product with probability $1- 0.01/k^2$, then by union bound we can compute the full product with probability $0.99$. This is the idea for the reduction.

For the reduction, we rely on the presence of good vectors. We say that a vector $v\in\F^n$ is good if $\Pr_{M, {\sf ALG}}[{\sf ALG}^{M, v} = Mv] \ge \alpha / 2$. By an averaging argument, we show that at least $\frac{\alpha}{2}|\F^n|$ good vectors exist in $\F^n$. Then, given input $M\in\F^{d\times n}$, by sampling $k$ independent random $M_i \in\F^{d\times n}$, we can apply the reduction idea to construct a $\F^{n\times n}$ matrix. Along with an idea by \cite{Blum_Luby_Rubinfeld_1993}, we get an algorithm for MvM where the matrix is of size $d\times n$. We then repeat this idea to split a $d\times n$ matrix into $k$ independent $d\times d$ matrices, giving us an algorithm for MvM that succeeds with high probability, where $M\in\F^{d\times d}, v\in\F^d$. So by boosting the success probability to $1- 0.01/k^2$, we can compute MvM for $M\in\F^{n\times n}, v\in \F^n$, with probability $0.99$.

In the quantum setting, the superposition of the matrix entries prevents us from naively applying unitaries to each sub-matrix. To overcome this, we will require some extra work for concatenating and extracting the sub-matrices and sub-vectors. For concatenation, we use a multiplexor circuit along with shifting of indices to get the correct combination of matrices. For extraction of matrices, we also have to be careful with the indices of the matrices and shift them accordingly to achieve the correct indexing. We specify in detail how to do this in section 4. 

\subsection{Comparison to Prior Work}

Compared to the additive combinatorics-based approach of~\cite{Asadi_Golovnev_Gur_Shinkar_Subramanian_2022}, our reduction avoids the use of complex Fourier-analytic and additive-combinatorial machinery. Compared to classical hardness self-amplification techniques~\cite{Hirahara_Shimizu_2023}, we show that the key ideas can be transplanted into the quantum world with only modest technical modifications.

We believe that the simplicity and modularity of our approach will make it easier to extend quantum worst-case to average-case reductions to new settings, and could help in designing fine-grained quantum reductions for problems of cryptographic or algorithmic significance.

\section{Preliminaries}
For $n\in\N$, we denote $[n] = \st{1, . . . , n}$. Also, we write $x \sim \mu$ to denote that $x$ is drawn from a distribution $\mu$. For a distribution $\mu$, we let $\text{supp}(\mu) = \st{x: \mu(x) > 0}$ denote the support of $\mu$.
\subsection{Union Bound}
In our proof, we will often apply the union bound to achieve lower bounds on our probability of success. We first state the lemma.
\begin{lemma}[Union Bound]
\label{union_bound}
For a countable set of events $A_1, A_2, ...$, we have
\[ \Pr\left( \bigcup_{i=1}^\infty A_i \right) \le \sum_{i=1}^\infty \Pr(A_i)\]
\end{lemma}
In our proof, we will use the union bound argument as follows. Suppose we have events $A_1, A_2, ..., A_n$ that each succeed with probability at least $1-\delta$. Then denoting $A_k'$ as the complement event of $A_k$, we have that $\Pr(A_k) \le \delta$ for all $k\in [n]$. So that 
\[ \Pr\left( \bigcap_{i=1}^n A_i \right) = 1- \Pr\left( \bigcup_{i=1}^n A_i' \right) \ge 1- \sum_{i=1}^n \Pr (A_i') \ge 1 - n\delta \]
\subsection{Quantum Unitary Oracles}

\textbf{Quantum Unitary Oracles.} We let quantum algorithms access a matrix $M\in\F^{n\times n}$ via a unitary oracle $U_M$ that does the following map
\[ U_M\ket{j, k, z} = \ket{j, k, z\oplus M_{jk}} \]
for all $j,k\in [n]$ and $z\in \F$, and $\oplus$ denotes addition over $\F$. Similarly for a vector $v\in\F^n$, oracle access means access to unitary $U_v$ such that
\[ U_v\ket{j, z} = \ket{j, z\oplus v_j} \]
for all $j\in [n]$ and $z\in \F$.

\subsection{Average-case quantum algorithms} 
An average-case quantum algorithm ALG is one that succeeds with probability at least $\alpha$ in expectation over (uniformly) random inputs. This probability is taken to be the probability of measuring the correct output on the output register in the standard basis. Specifically, we define it as follows.
\begin{defn} 
Suppose that ${\sf ALG}$ computes some function $f$ over some known measurable domain $V$, then we define the probability of success as follows:
\[ \Pr_{\substack{\textsf{ALG} \\ v\in V}}[\textsf{ALG}(v) = f(v)] := \E_{v\in V}[ |\Pi_{f(v)}\textsf{ALG}(v)|^2 ] \ge \alpha \]
where $\Pi_{f(v)}$ denotes an orthogonal projection on the output register of $\sf ALG$ which represents measuring the outcome $f(v)$ in the standard basis. 
Note that this probability is over the internal quantum randomness and the inputs, which is highlighted by the notation used.
\end{defn}
With this definition in mind, we now define what a noisy quantum oracle means.
\\\\
\textbf{Noisy quantum oracles.} For a function $f:\F^n \to \F^m$, we can consider a unitary oracle $U_f$ which on input $x\in \F^n, z\in \F^m$ performs the map
\[ U_f\ket{x}\ket{z} = \beta_{succ}^x \ket{x}\ket{z +f(x)} + \beta_{fail}^x \ket{x}\ket{\psi(x)} \]
where $\beta_{succ}^x, \beta_{fail}^x\in\C$ such that $|\beta_{succ}^x|^2 + |\beta_{fail}^x|^2 = 1$, and the normalized state $\ket{\psi (x)}$ could be an arbitrary superposition over $\ket{z+v}$ for $v\in \F^m \setminus \st{f(x)}$, and $+$ denotes component-wise addition for vectors over $\F$. We can think of $U_f$ as a quantum analogue of a classical probabilistic algorithm for computing $f$, which outputs the correct value of $f(x)$ with probability $|\beta_{succ}^x|^2$ when measuring the second register in the computational basis.
\subsection{Sampler}
Consider a edge-weighted bipartite graph $Q=(X,Y,W)$, whose weights are given by $W \in [0,1]^{X\times Y}$. Assume that $W$ is normalized such that $\sum_{x,y} W(x,y) = 1$, which specifies a distribution over $X\times Y$. This also induces the marginal distributions $\mu \in [0,1]^X, \nu\in [0,1]^Y$ as such:
\[ \mu(x) = \sum_{y\in Y} W(x,y) ,\quad \nu(y) = \sum_{x\in X} W(x,y) \]
Consider a measure $w:Y\to [0,1]$. We define the expectation of $w$ as follows:
\[ \E [w] = \sum_{y\in Y} \nu(y) w(y) \]
We say that a measure $w\in[0,1]^Y$ is $\veps$-dense if $\E[w] \ge \veps$. Similarly, we can define it for $X$. We also say a set $U\subset Y$ is $\veps$-dense if $\Pr[Y\in U] \ge \veps$ with respect to the distribution $\nu$.
\begin{defn}(Sampler)
    For $0<c<\delta$, $Q=(X,Y,W)$ is a $(\delta, c)$-sampler for density $\veps$, if for any $\veps$-dense measure $w\in[0,1]^Y$, we have
    \[ \Pr_{x\sim \mu}\bigg[ \E[w(Y) |X=x] \le \E[w(Y)](1-c) \bigg]\le \delta \]
\end{defn}
In our case, our reduction is done via uniform random sampling, and thus we may consider the simplified definition for unweighted graphs.
\begin{defn}(Sampler, unweighted)
    For $0<c<\delta$, $Q=(X,Y,W)$ is a $(\delta, c)$-sampler for density $\veps$, if for any $\veps$-dense subset $U\subset Y$, we have
    \[ \Pr_{x\sim \mu}\bigg[ \Pr[Y\in U | X=x] \le \Pr[Y\in U](1-c) \bigg]\le \delta \]
\end{defn}
Intuitively, a sampler allows us to, for almost all $x$, via the reduction, land within a chosen $\veps$-dense subset $U$ with high probability. This will be used in the next section on hardness self-amplification.

\subsection{Hardness Self-Amplification}
\cite{Hirahara_Shimizu_2023} developed a framework for hardness self-amplification, which we will be using for the reduction. We state the main results we will be using from their paper. For a distributional problem $(f, \mu)$ with $X = \text{supp}(\mu)$ and $k \in \mathbb{N}$, consider the $k$-wise direct product $(f^k, \mu^k)$, where
\[
f^k(x_1, \ldots, x_k) = \big(f(x_1), \ldots, f(x_k)\big) \quad \text{and} \quad \mu^k(x_1, \ldots, x_k) = \prod_{i \in [k]} \mu(x_i).
\]
Let $Y = X^k$ and $\mathcal{O}$ be an oracle for $(f^k, \mu^k)$. Consider the following well-known reduction $\mathcal{R}^\mathcal{O}$ from $(f, \mu)$ to $(f^k, \mu^k)$ for the direct product problem (\cite{Impagliazzo_Wigderson_1997, Impagliazzo_1995, Goldreich_Nisan_Wigderson_2011}). Given input $x \in X$, 
\begin{itemize}
    \item Sample $i \in [k]$ and $(x_1, \ldots, x_k) \sim \mu^k$.
    \item Let $\bar{x} = (x_1, \ldots, x_{i-1}, x, x_{i+1}, \ldots, x_k)$ be a query.
    \item Given $(y_1, \ldots, y_k) = O(\bar{x})$, output $y_i$.
\end{itemize}
Then, via this reduction, \cite{Hirahara_Shimizu_2023} proved the following results which are central to our matrix-vector multiplication reduction.
\begin{lemma}[Direct Product Lemma]
    \label{lma:DirectProductLemma}
    The query graph of the reduction $\mathcal{R}$ is a $(\delta, c)$-sampler for density $\veps$ for any $\veps, \delta, c$ that satisfy $2\exp{(-kc^2\delta /8)}\le c\veps$.
\end{lemma}
\begin{lemma}[Direct Product Theorem for Verifiable Function]
    \label{lma:1.3}
    Let $(f, \mu)$ be a distributional problem such that $f$ is $t(n)$-time verifiable. Let $k \in \mathbb{N}, \delta > 0, \epsilon > 0$ be such that $\epsilon \geq 4 \exp(-\delta k / 32)$. If there exists a $T(n)$-time algorithm $\mathcal{M}$ for $(f^k, \mu^k)$ with success probability $\epsilon$, then there exists an $O(\epsilon^{-1}(T(n) + t(n)))$-time randomized algorithm $\mathcal{M}'$ that solves $(f, \mu)$ with success probability $1 - \delta$.
\end{lemma}
In particular, \Cref{lma:1.3} provides us with the technique for the reduction. Suppose we have a verifier for our problem. Then, by choosing $k=O(\log (1/\veps)$, and performing the reduction $\mathcal{R}$, we can amplify the success rate of our algorithm to $1-\delta$ for a chosen $\delta$.

\subsection{Matrix-vector Product Verification (MvPV)}
For the Direct Product Theorem to work, we need a matrix-vector product verifier. Through fixed-point amplification, \cite{Asadi_Golovnev_Gur_Shinkar_Subramanian_2022} showed a circuit for flagging matrix-vector product in superposition with $O(n^{3/2}\cdot\log(\frac{1}{\veps}))$ queries to the matrix. Specifically, they proved the following lemma.
\begin{lemma}[Noisy quantum MvPV]
\label{lma:MvPV}
Suppose we are given a quantum oracle $U_M$ for a matrix $M \in \F^{n \times n}$ 
\[
U_M \, \ket{j, k, z } = \ket{j, k, z \oplus M_{jk}}
\]
for all indices $j, k \in [n]$ and $z \in \F$, and a noisy quantum algorithm $\text{\sf ALG}$, i.e.,
\[
\text{\sf ALG} \, \ket{v} = \beta^\nu_{\text{succ}} \, \ket{v} \ket{Mv} \ket{w_0(v)} + \beta^\nu_{\text{fail}} \ket{v} \ket{\Psi(v)}.
\]
Then there exists a gate-efficient quantum algorithm $\text{\sf ALG}_{\text{\sf verified}}$ that succeeds and outputs $Mv$ with probability $|\beta^\nu_{\text{succ}}|^2$ along with a flag indicating success, and similarly outputs a flag indicating failure whenever it outputs a vector $z \neq Mv$, with probability at least $(1 - \varepsilon)|\beta^\nu_{\text{fail}}|^2$. Furthermore, $\text{\sf ALG}_{\text{\sf verified}}$ uses $q = O\left(n^{3/2} \cdot \log \frac{1}{\varepsilon}\right)$ queries to $U_M$ and $U_M^\dagger$, $O(1)$ queries to $\text{\sf ALG}$, $O\left(q \log n \cdot \text{poly} \log |\mathbb{F}|\right)$ additional one-qubit and two-qubit gates, and $O(n \log |\mathbb{F}|)$ ancillary qubits.
\end{lemma}
\noindent
This provides us with a verifier for our algorithm, by measuring the flag qubit after performing the procedure mentioned in \Cref{lma:MvPV}.

\section{Concatenation of Matrices}
In our reduction, we will have to operate on subvectors and submatrices in order to compute a solution for a subvector. Then, we will later concatenate these subvectors to form the full vector solution for the matrix-vector product. To do this, we need a process to concatenate and extract vectors.

\subsection{Concatenation of matrices via linear combination of unitaries}
We use a multiplexor circuit \cite{Bergholm_Vartiainen_Mottonen_Salomaa_2004} to build a circuit allowing us to achieve a uniform linear combination of unitaries. 
Specifically, consider the circuit:
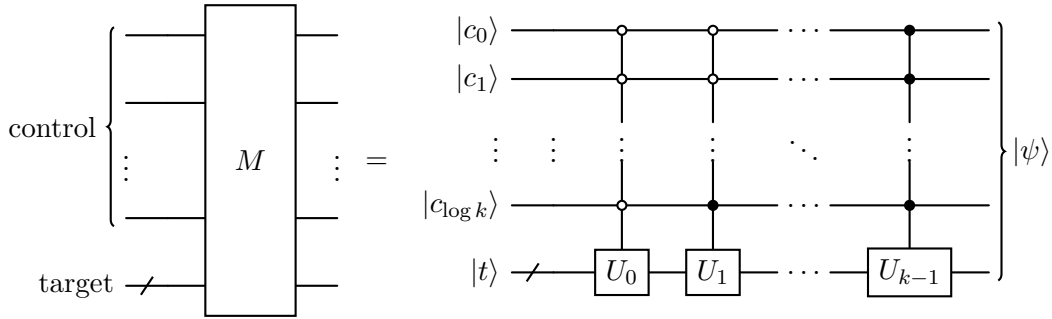
\begin{center}
    \begin{quantikz}
        \lstick[4]{control}&&\gate[5][1.2cm][0.8cm]{M} &\\
        &&&\\
        \setwiretype{n}\vdots&&&\vdots\\
        &&&\\
        \lstick{target}&\qwbundle{}&&
    \end{quantikz}\, = \,
    \begin{quantikz}
    \lstick{\ket{c_0}} &  & \octrl{1} & \octrl{1} & \ \dots \ & \ctrl{1} & \rstick[5]{\ket{\psi}} \\
    \lstick{\ket{c_1}} &  & \octrl{1} & \octrl{1} & \ \dots \ & \ctrl{1} &  \\
    \lstick{\vdots}  & \setwiretype{n}\vdots & \vdots  & \vdots & \ddots & \vdots &  \\
    \lstick{\ket{c_{\log k}}} &  & \octrl{1}\wire[u]{q} & \ctrl{1}\wire[u]{q} & \ \dots \ & \ctrl{1}\wire[u]{q} & \\
    \lstick{\ket{t}} &\qwbundle{} & \gate{U_0} & \gate{U_1} & \ \dots \ & \gate{U_{k-1}} & 
    \end{quantikz}
\end{center}
where we have $\log k$ control qubits and some output register(s) $\ket{t}$. When the control qubits are a uniform superposition from $0$ to $k-1$, applying the multiplexor $M$ gives us the output
\[ \ket{\psi} = \frac{1}{\sqrt{k}}\sum_{i=0}^{k-1} \ket{i}U_k\ket{t} \]
With this in mind, we now describe a procedure that allows us to concatenate vectors and matrices in the quantum setting. 
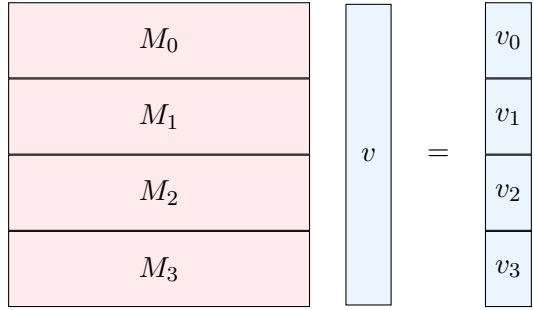
\begin{figure}[ht]
\centering
    \begin{tikzpicture}
        \definecolor{matrixcolor}{RGB}{255, 235, 235} 
        \definecolor{vectorcolor}{RGB}{235, 245, 255} 
    
        \matrix[matrix of math nodes, nodes={draw, minimum height=1cm, minimum width=4cm, fill=matrixcolor}, 
                row sep=0pt] (m) {
            M_0 \\
            M_1 \\
            M_2 \\
            M_3 \\
        };
        \matrix[matrix of math nodes, nodes={draw, minimum height=4cm, minimum width=0.6cm, fill=vectorcolor}, 
                right=0.2cm of m] (v1) {
            v \\
        };
    
        \node[right=0.2cm of v1] (label) {$ =$};
    
        \matrix[matrix of math nodes, nodes={draw, minimum height=1cm, minimum width=0.6cm, fill=vectorcolor}, 
                right=0.2cm of label] (v) {
            v_0 \\
            v_1 \\
            v_2 \\
            v_3 \\
        };
    
    \end{tikzpicture}
    \caption{Concatenation of submatrices to compute the full matrix-vector product.}
    \label{fig:1}
\end{figure}
First, recall that we let quantum algorithms access a matrix $M\in\F^{n\times n}$ via a unitary $U_M$ that does the following map
\[ U_M\ket{j, k, z} = \ket{j, k, z\oplus M_{jk}} \]
for all $j,k\in [n]$ and $z\in \F$, and $\oplus$ denotes addition over $\F$. Now, suppose that we have oracle access to matrices $M_0, ..., M_{k-1} \in \F^{d\times n}$ via unitaries $U_{M_0}, ..., U_{M_{k-1}}$, and that $k$ divides $n$ and $d = n/k$. We would like to construct a unitary $U_M$ for $M\in\F^{n\times n}$, the row-wise concatenation of $M_0, ..., M_{k-1}$, like shown in \hyperref[fig:1]{Figure 1}. By fixing $d$ as above, consider the floor division function $q : [n] \to [k]$ given by
\[ q(z) = \bigg\lfloor \frac{z}{d} \bigg\rfloor  \]
We can build a quantum circuit $U_{q}$ that maps $\ket{z}\ket{0} \to \ket{z}\ket{q(z)}$ for every $z\in\F$. Furthermore, using a quantum arithmetic circuit, we can build a unitary $U_r$ that maps $\ket{z}\ket{q(z)}\ket{0} \to \ket{z}\ket{q(z)}\ket{z - d\cdot q(z)}$. Putting all these together, we have the following:
\begin{lemma}
    \label{concat_lemma}

    Suppose we are given as input $N$ unitaries $U_0, U_1, ..., U_{N-1}$ that give oracle access to matrices $M_0, M_1, ..., M_{N-1} \in \F^{d\times n}$, i.e.
    \[ U_k\ket{i,j ,z} = \ket{i,j, z \oplus (M_k)_{i,j}} \]
    for all $k\in [N]$, and $i\in[d], j\in[n]$. Denote $m = Nd$. Then there is a unitary $U$ that makes 1 query to each unitary $U_i$, and gives oracle access to $M\in\F^{m\times n}$, the row-wise concatenation of the matrices $M_0, M_1, ..., M_{N-1}$, so that for all $i \in [m], j \in [n]$,
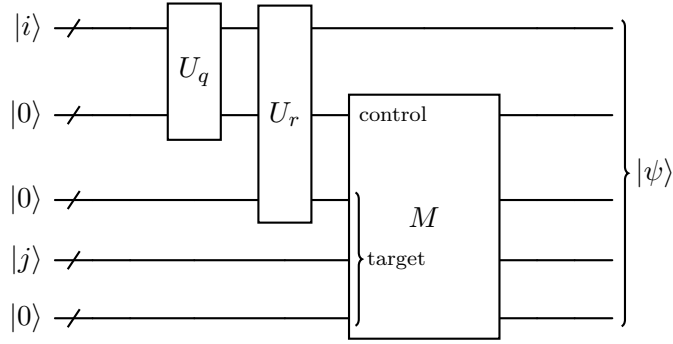
    \[ U\ket{i,j ,z} = \ket{i, j, z \oplus (M_{\lfloor i/d \rfloor})_{(i\bmod d),j}} \]
\end{lemma}
\begin{center}
    \begin{quantikz}
    \lstick{\ket{i}} & \qwbundle{} && \gate[2]{U_{q}} & \gate[3]{U_r} &  &  &  & \rstick[5]{\ket{\psi}} \\
    \lstick{\ket{0}} & \qwbundle{} &&  &  &  \gate[4][2cm]{M}\gateinput{control} & &  &  \\
    \lstick{\ket{0}} & \qwbundle{} &&  &   &\gateinput[3]{target}  &  &  &  \\
    \lstick{\ket{j}} & \qwbundle{} &&  &  &  & &  & \\
    \lstick{\ket{0}} & \qwbundle{} &&   &  &  &  &  & 
    \end{quantikz}
\end{center}
where the multiplexor $M$ has unitaries $U_i$ for all $i\in[N]$.
\begin{proof}
The input state undergoes the following transformations through the circuit:
\begin{align*}
    \ket{i}\ket{0}\ket{0}\ket{j}\ket{0} &\xrightarrow{U_{q} \otimes \mathbbm{1} \otimes \mathbbm{1} \otimes \mathbbm{1}} \ket{i}\ket{ \lfloor i/d \rfloor}\ket{0}\ket{j}\ket{z} \\
    &\xrightarrow{U_r \otimes \mathbbm{1} \otimes \mathbbm{1}} \ket{i}\ket{ \lfloor i/d \rfloor}\ket{i\bmod{d}}\ket{j}\ket{z}  \\
    &\xrightarrow{\mathbbm{1} \otimes M} \ket{i}\ket{ \lfloor i/d \rfloor}\ket{i\bmod{d}}\ket{j}\ket{z\oplus (M_{\lfloor i/d \rfloor})_{(i\bmod d),j}} 
\end{align*}
By looking at the first register and the last 2 registers, we achieve what we want:
a concatenation of submatrices with the new row indexing. This circuit uses 1 query to each unitary $U_i$ in the multiplexor.
\end{proof}
With a row-wise concatenation procedure, we can similarly build a column-wise concatenation procedure. We state it as an immediate corollary, whose proof and circuit follows similarly from \Cref{concat_lemma}.
\begin{corollary}
\label{concat_lemma_col}
    Suppose we are given as input $N$ unitaries $U_0, U_1, ..., U_{N-1}$ that give oracle access to matrices $M_0, M_1, ..., M_{N-1} \in \F^{n\times d}$, i.e.
    \[ U_k\ket{i,j ,z} = \ket{i,j, z \oplus (M_k)_{i,j}} \]
    for all $k\in [N]$, and $i\in[n], j\in[d]$. Denote $m = Nd$. Then there is a unitary $U$ that makes 1 query to each unitary $U_i$, and gives oracle access to $M\in\F^{n\times m}$, the column-wise concatenation of the matrices $M_0, M_1, ..., M_{N-1}$, so that for all $i\in [n], j \in [m]$,
    \[ U\ket{i,j ,z} = \ket{i, j, z \oplus (M_{\lfloor j/d \rfloor})_{i,(j\bmod d)}} \]
\end{corollary}

\subsection{Extraction of matrices and vectors}
For extraction of matrices and vectors, note that given a unitary $U_M$ for a matrix $M$, it is easy to retrieve a submatrix by simply applying the unitary on the desired indices. Specifically, quantum algorithms access a matrix $M\in\F^{n\times n}$ via a unitary that does the following map
\[ U_M\ket{j, k, z} = \ket{j, k, z\oplus M_{j,k}} \]
so by simply performing of shift of indices, we can get a unitary for the submatrix. Specifically, we have the following circuit for extracting a $d\times n$ submatrix from an $n\times n$ matrix. We wish to use an adder circuit $ADD_{p}$ defined by
\[ ADD_p \ket{z} = \ket{z \oplus p} \]
where $\oplus$ denotes addition over $\F$. Implementing such a gate requires only O(poly log $|\F|$) elementary two-qubit gates \cite{beauregard2003quantumarithmeticgaloisfields}.

\begin{lemma}
    \label{extraction_lemma}
    Suppose that we have oracle access to a matrix $M\in\F^{n\times n}$ via a unitary $U_M$. Then for any $0< d \le n$ and any $0\le p \le n-d$, there exists a unitary $U_{M'}$ that makes 1 query to $U_M$, O(poly log$|\F|$) elementary two-qubit gates, and provides oracle access to the submatrix $M'\in \F^{d\times n}$, where $M'_{i,j} = M_{p+i, j}$ for all $i\in[d],j\in[n]$. In other words,
    \[ U_{M'}\ket{i,j,z} = \ket{i, j,z \oplus M_{p+i, j}} \]
    for all $i\in [d], j\in[n]$.
\end{lemma}
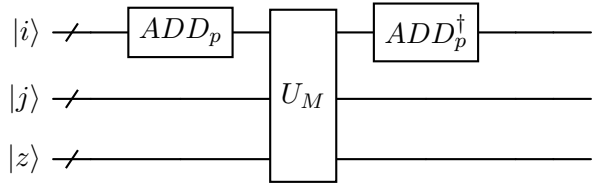
\begin{center}
    \begin{quantikz}
    \lstick{\ket{i}} & \qwbundle{} & \gate[]{ADD_p} & \gate[3]{U_M} & \gate[]{ADD_p^\dagger} &  &  &  \\
    \lstick{\ket{j}} & \qwbundle{} &  &   &  &  &  &  \\
    \lstick{\ket{z}} & \qwbundle{} &   &  &  &  &  & 
    \end{quantikz}
\end{center}
\begin{proof}
    The circuit undergoes the following transformations through the circuit:
    \begin{align*}
        \ket{i,j,z} &\xrightarrow{ADD_p \otimes \mathbf{1}} \ket{i \oplus p, j, z} \\
        &\xrightarrow{U_M} \ket{i\oplus p, j, z \oplus M_{p+i, j}} \\
        &\xrightarrow{ADD_p^\dagger \otimes \mathbf{1}} \ket{i, j, z \oplus M_{p+i, j}} 
    \end{align*}
\end{proof}
Similarly, we can perform the same procedure for vectors, which is simpler since there is only 1 dimension of indices. We state it as an immediate corollary.
\begin{corollary}
    \label{extraction_lemma_vector}
    Suppose that we have oracle access to a vector $v\in\F^{n}$ via a unitary $U_v$. Then for any $0< d \le n$ and any $0\le p \le n-d$, there exists a unitary $U_{v'}$ that makes 1 query to $U_v$, O(poly log$|\F|$) elementary two-qubit gates, and provides oracle access to the subvector $v'\in \F^{d}$, where $v'_{i} = M_{p+i}$ for all $i\in[d]$. In other words,
    \[ U_{v'}\ket{i,z} = \ket{i,z \oplus v_{p+i}} \]
    for all $i\in [d]$.
\end{corollary}
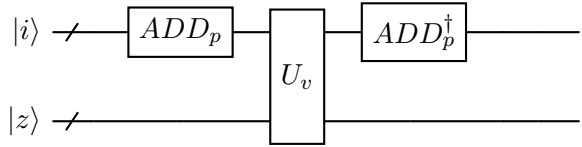
\begin{center}
    \begin{quantikz}
    \lstick{\ket{i}} & \qwbundle{} & \gate[]{ADD_p} & \gate[2]{U_v} & \gate[]{ADD_p^\dagger} &  &  &  \\
    \lstick{\ket{z}} & \qwbundle{} &   &  &  &  &  & 
    \end{quantikz}
\end{center}
\section{Reduction for Matrix-vector Multiplication}
The reduction for matrix-vector multiplication follows closely to \cite{Hirahara_Shimizu_2023}, with additional care taken to deal with concatenation and extraction of matrices.
\subsection{Proof of \texorpdfstring{\Cref{thm:main}}{\ref{thm:main}}}
In this section, we work towards proving \Cref{thm:main}. First, we say that a vector $v\in\F^n$ is good if $\Pr_{M, {\sf ALG}}[{\sf ALG}^{M, v} = Mv] \ge \alpha / 2$. Let $k$ be such that $n$ is divisible by $k$ and that $\alpha \ge 4\exp(-k/3200)$ so that the condition in \Cref{lma:1.3} is satisfied for $\delta=0.01$. Let $d = n/k$. Throughout this section, suppose that there exists a quantum algorithm {\sf ALG} that has oracle access to a matrix $M\in \F^{n\times n}$ and a vector $v$, makes $Q(n)$ queries, and satisfies
\[
\Pr_{\substack{M, v, \\ \text{\sf ALG}}} \left[ \text{\sf ALG}^{M, v} = Mv \right] 
= \underset{\substack{M \in \mathbb{F}^{n \times n} \\ v \in \mathbb{F}^n}}{\mathbb{E}}
\left[ \left\| \Pi_{Mv} \text{\sf ALG}^{M, v} \ket{0} \right\|^2 \right] \geq \alpha \,
\]
\begin{lemma}
\label{lma:1.5}
    At least $\frac{\alpha}{2}|\F^n|$ good vectors exists in $\F^n$.
\end{lemma}

\begin{proof}
    Let $p_v = \Pr_{M, {\sf ALG}}[{\sf ALG}^{M,v} = Mv]$. By our assumption of the average-case algorithm, we have $\E_{v}[p_v] \ge \alpha$. Then, note that
    \[ \alpha \le \E_{v}[p_v] \le 1\cdot \Pr_v[p_v \ge \alpha/2] + \alpha/2 \cdot\Pr_v[p_v < \alpha/2] \le \Pr_v[p_v \ge \alpha/2] + \alpha/2 \]
    It follows that $\Pr_v[p_v \ge \alpha/2] \ge \alpha/2$ as desired.
\end{proof}

\begin{lemma}
    \label{lma:1.6}
    Let $v\in\F^n$ be good, and $n$ divisible by $k$. Then there exists a quantum algorithm ${\sf ALG}_1$ that makes $O(\alpha^{-1}\cdot(Q(n)+n^{3/2}))$ queries to $U_M$ and satisfies
\[
\Pr_{M \sim \F^{d \times n}} \left[ \Pr_{{\sf ALG}_1} \left[ {\sf ALG}_1^{M, v} = Mv \right] \geq 0.99 \right] \geq 0.99 
\]
\end{lemma}

\begin{proof}
    For a given matrix $M\in\F^{d\times n}$, sample $k$ independent matrices $M_1, ..., M_k \sim \F^{d\times n}$ and index $i \sim [k]$. Construct a random matrix $\overline{M}\in\F^{n\times n}$ by aligning the matrices $M_1, ..., M_{i-1}, M, M_{i+1},..., M_k$ via \Cref{concat_lemma}, and denote its unitary by $U_{\overline{M}}$. Compute $w ={\sf ALG}^{\overline{M}, v}$. Let $w_1, ..., w_k$ be the subvector of $w$ where each $w_i\in\F^d$ consists of the $((i-1)d + 1)$-st to $id$-th elements of $w$. Then, use
    \Cref{lma:MvPV} to verify whether $\overline{M}v = w$. If so, return $w_i$ via \Cref{extraction_lemma_vector}. Repeat this whole procedure $O(1/\alpha)$ times. The correctness of $\sf ALG_1$ follows from \Cref{lma:1.3}. By \Cref{concat_lemma}, $U_{\overline{M}}$ uses 1 call to $U_M$ in each iteration, giving a total of $O(\alpha^{-1})$ queries to $U_M$, while \Cref{lma:DirectProductLemma} gives $O(\alpha^{-1}\cdot(Q(n)+n^{3/2}))$ queries to $U_M$.
\end{proof}

\begin{corollary}
\label{cor:1.7}
    Let $v\in\F^n$ be good, and $n$ divisible by $k$. Let $d = n / k$. Then for every $M\in\F^{d\times n}$, there exists a quantum algorithm ${\sf ALG}_2$ that makes $O(\alpha^{-1}\cdot(Q(n)+n^{3/2}))$ queries to $U_M$
\[
\Pr_{{\sf ALG}_2} \left[ {\sf ALG}_2^{M, v} = Mv \right] \geq 0.96
\]
\end{corollary}

\begin{proof}
    Given a matrix $M\in \F^{d\times n}$ and $v\in\F^n$, sample $R_1 \sim \F^{d\times n}$ uniformly randomly, and let $R_2 = M - R_1$. Using \Cref{lma:1.6}, compute ${\sf ALG_1}^{R_1,v}$ and ${\sf ALG_1}^{R_2,v}$, and denote their unitaries as $U_1$ and $U_2$ respectively. For the algorithm to succeed, we need 4 events to succeed: $R_1$ and $R_2$ each satisfying $\Pr_{\sf ALG_1}[{\sf ALG}_1^{M,v}= Mv] \ge 0.99$, and that both products $R_1v$ and $R_2v$ are computed correctly. All 4 events each have success probability $0.99$ by \Cref{lma:1.6}. By union bound (\Cref{union_bound}), we can obtain $U_1$ and $U_2$ with success probability $0.96$. By using the circuit below, we can compute ${\sf ALG_1}^{R_1,v} + {\sf ALG_1}^{R_2,v}$.
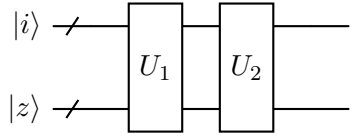
    \begin{center}
        \begin{quantikz}
        \lstick{\ket{i}} & \qwbundle{} & \gate[2]{U_{1}} & \gate[2]{U_{2}} &  &  \\
        \lstick{\ket{z}} & \qwbundle{} &   &  &  & 
        \end{quantikz}
    \end{center}
    where $\ket{i}$ is our row index register. The circuit undergoes the following transformation
    \begin{align*}
        \ket{i, z} &\xrightarrow{U_{1}} \ket{i, z \oplus (R_1)_i} \\
        &\xrightarrow{U_{2}} \ket{i, z \oplus (R_1v)_i \oplus (R_2v)_i}  = \ket{i, z \oplus (Mv)_i}
    \end{align*}
    
\end{proof}
\begin{lemma}
    \label{lma:1.8}
    Let $n$ divisible by $k$ and $d = n / k$. Then for every matrix $M\in\F^{d\times d}$. there exists a quantum algorithm ${\sf ALG}_3$ that makes $O(\alpha^{-2}\cdot(Q(n)+n^{3/2}))$ queries to $U_M$ and satisfies
\[
\Pr_{v \sim \F^{d}} \left[ \Pr_{{\sf ALG}_3 } \left[ {\sf ALG}_3^{M, v} = Mv \right] \geq 0.99 \right] \geq 0.99 
\]
\end{lemma}
\begin{note}
    {$\sf ALG_3$} computes $Mv$ for a smaller matrix $M\in\F^{d\times d}$ and $v\in\F^d$.
\end{note}
\begin{figure}[ht]
\centering
    \begin{tikzpicture}
        \definecolor{matrixcolor}{RGB}{255, 235, 235} 
        \definecolor{vectorcolor}{RGB}{235, 245, 255} 
    
        \matrix[matrix of math nodes, nodes={draw, minimum height=1.2cm, minimum width=1cm, fill=matrixcolor}, 
               row sep=0pt] (m) {
            0 & M & 0 & 0 \\
        };
        \matrix[matrix of math nodes, nodes={draw, minimum height=1cm, minimum width=0.7cm, fill=vectorcolor}, 
                right=0.2cm of m] (v1) {
            v_1 \\ v \\ v_3 \\ v_4 \\
        };
    
        \node[right=0.2cm of v1] (label) {$ =$};
    
        \matrix[matrix of math nodes, nodes={draw, minimum height=1.2cm, minimum width=0.7cm, fill=vectorcolor}, 
                right=0.2cm of label] (v) {
            Mv \\
        };
    
    \end{tikzpicture}
    \caption{$M'\overline{v} = Mv$, where $M\in \F^{d\times d}$ and $v\in\F^d$ is embedded in the second block i.e. $i=2$.}
    \label{fig:2}
\end{figure}
\begin{proof}
    Let vector $v\in\F^d$ be given. Sample $v_1, ..., v_k \sim \F^d$ and $i \sim [k]$ and construct $\overline{v} = (v_1, ..., v_{i-1}, v, v_{i+1}, ..., v_k) \in \F^n$ via \Cref{concat_lemma}. Let $S\subset (R^d)^k$ be the set of $k$-tuple of vectors $(v_1, ..., v_k)$ such that the concatenation $(v_1,..., v_k)\in\F^n$ is good. From \Cref{lma:1.5}, $|S| \ge \frac{\alpha}{2}|\F^d|^k$. Let $H= \st{ v\in\F^d : \Pr[\overline{v} \in S] \ge 0.25\alpha}$. From \Cref{lma:DirectProductLemma}, given our choice of $k$, we have that the query graph is a $(0.01, 0.5)$-sampler and since $S\subset \F^n$ is an $\alpha/2$-dense subset, we have that
    \[ \Pr_{v\in\F^n}\left[ \Pr[\overline{v}\in S ]\le \frac{\alpha}{2}\cdot \frac{1}{2} \right] \le 0.01 \]
    Rearranging gives us that
    \[ \Pr_{v\in\F^n}\left[ v\in H \right] = \Pr_{v\in\F^n}\left[ \Pr[\overline{v}\in S]\ge \frac{\alpha}{4} \right] \ge 0.99 \]
    so that $|H| \ge 0.99|\F^d|$. Suppose $v$ is embedded in the $i$-th block of $\overline{v}$. Let $M' \in \F^{d\times n}$ be the matrix such that the $i$-th block is $M$ and the other elements are set to zero, i.e. $M'_{[d], I_i} = M$, where $I_i = \st{(i-1)d + 1, ..., id}$ (see \Cref{fig:2}). By \Cref{concat_lemma_col}, we can construct a unitary for $M'$. Then, using $\sf ALG_2$, compute $w = {\sf ALG_2}^{M', \overline{v}}$ and use \Cref{lma:MvPV} to verify whether $M', \overline{v} = Mv$. If so, return $w$ via \Cref{extraction_lemma_vector}. Otherwise, sample $\overline{v}$ again and repeat.
    \\
    \\
    For correctness, note that $v\in H$ with probability 0.99 if $v$ is chosen uniformly random. Furthermore, if $v\in H$, then $\overline{v}$ is good with probability $0.25\alpha$. If $\overline{v}$ is good, $\sf ALG_2$ computes $M'\overline{v}$ in $O(\alpha^{-1}\cdot(Q(n)+n^{3/2}))$ queries to $U_M$. Together, we have that with probability 0.99 over the choice of $\overline{v}$, the algorithm will terminate within $O(1/\alpha)$ iterations. In each iteration, the unitary for $M'$ makes 1 query to $U_M$. We also make 1 query to $\sf ALG_2$ and the verifier. Together, we get that $\sf ALG_3$ makes $O(\alpha^{-2}\cdot(Q(n)+n^{3/2}))$ queries to $U_M$.
\end{proof}
\begin{corollary}
    \label{cor:1.10}
    Let $n$ divisible by $k$ and $d = n / k$. Then for every $v\in\F^d$ and $M\in\F^{d\times d}$, there exists a quantum algorithm ${\sf ALG}_4$ that makes $O(\alpha^{-2}\cdot(Q(n)+n^{3/2}))$ queries to $U_M$ and satisfies
\[
 \Pr_{{\sf ALG}_4}  \left[ {\sf ALG}_4^{M, v} = Mv \right] \geq 0.96
\]
\end{corollary}
\begin{proof}
    Similar to \Cref{cor:1.7}, given $v\in\F^d$ and $M\in\F^{d\times d}$, ${\sf ALG_4}$ samples $r_1\sim \F^d$ and let $r_2 = v - r_1$. Then using \Cref{lma:1.8}, compute ${\sf ALG_3}^{M,r_1} + {\sf ALG_3}^{M, r_2}$ similarly to \Cref{cor:1.7}. By union bound (\Cref{union_bound}), ${\sf ALG_4}$ succeeds with probability 0.96, giving us $Mr_1 + Mr_2 = Mv$.
\end{proof}
\begin{thm}
\label{thm:main}
    Let $\F = \F_p$ be a prime field, $n\in\N$, and $\alpha := \alpha(n) \in (0, 1]$. Suppose that there exists a quantum algorithm {\sf ALG} that has oracle access to a matrix $M\in \F^{n\times n}$ and a vector $v$, makes $Q(n)$ queries, and satisfies
\[
\Pr_{\substack{M, v, \\ \text{\sf ALG}}} \left[ \text{\sf ALG}^{M, v} = Mv \right] 
= \underset{\substack{M \in \mathbb{F}^{n \times n} \\ v \in \mathbb{F}^n}}{\mathbb{E}}
\left[ \left\| \Pi_{Mv} \text{\sf ALG}^{M, v} \ket{0} \right\|^2 \right] \geq \alpha \, .
\]
Then, for every constant $\delta > 0$, there exists a worst-case algorithm $\sf ALG'$ that makes $O(\alpha^{-2})$ queries to $\sf ALG$, makes $O((Q(n)+n^{3/2})\cdot \alpha^{-2}\cdot\log ^2(1/\alpha) \log\log(1/\alpha) )$ queries to $U_M$ and $U_M^\dagger$, and succeeds on all inputs with high probability:
\[
\forall M\in\F^{n\times n}, v\in\F^n, \Pr_{\text{\sf ALG}'}\left[(\text{\sf ALG}')^{M,v} = Mv\right] = 
\left\| \Pi_{Mv} (\text{\sf ALG}')^{M,v} |0\rangle \right\|^2 \geq 1 - \delta.
\]
\end{thm}

\begin{proof}
    Let $k=O(\log (1/\alpha))$. Suppose $k$ divides $n$ (by padding), and let $d = n/k$. For a given $M\in\F^{n\times n}$ and $v\in\F^n$, divide them into $M^{(i,j)}\in \F^{d\times d}$ via \Cref{extraction_lemma} and $v^{(i)}\in \F^d$ via \Cref{extraction_lemma_vector}, as shown in \Cref{fig:3}. Formally, let $I_i = \st{(i-1)d+1,..., id}$ and denote $M^{(i,j)} = M_{I_i, I_j}$ and $v^{(i)} = v_{I_i}$. For each $i, j\in [k]$, using \Cref{cor:1.10}, we can run $\sf ALG_{4}^{M^{(i,j)}, v^{(j)}}$ for $O(\log k)$ times to obtain $M^{(i,j)}v^{(j)}$ with probability $1-0.01/k^2$. Denote the unitary for each $M^{(i,j)}v^{(j)}\in \F^d$ as $U_{i,j}$. Then, by union bound (\Cref{union_bound}), with probability 0.99, we can obtain $M^{(i,j)} v^{(j)}$ for all $i,j\in[k]$ via $k^2$ unitaries. Then, for each $i\in[k]$, we can obtain a unitary for the vector $(Mv)^{(i)} = \sum_{j\in [k]} M^{(i,j)} v^{(j)}$ by simply applying each unitary $U_{i,j}$ for all $j\in [k]$. 
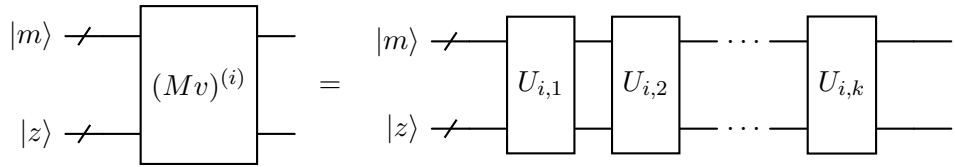
    \begin{center}
        \begin{quantikz}
        \lstick{\ket{m}}&\qwbundle{}&\gate[2][1.2cm][0.8cm]{(Mv)^{(i)}} &\\
        \lstick{\ket{z}}&\qwbundle{}&&
        \end{quantikz}\, = \,
        \begin{quantikz}
        \lstick{\ket{m}} & \qwbundle{} & \gate[2]{U_{i,1}} & \gate[2]{U_{i,2}} & \ \dots \ & \gate[2]{U_{i,k}} &  &  \\
        \lstick{\ket{z}} & \qwbundle{} &   &  & \ \dots \ &  &  & 
        \end{quantikz}
    \end{center}
    where $\ket{m}$ is our row index register. The circuit undergoes the following transformation
    \begin{align*}
        \ket{m, z} &\xrightarrow{U_{i,1}} \ket{m, z \oplus (M^{(i,1)}v^{(1)})_m} \\
        &\xrightarrow{U_{i,2}} \ket{m, z \oplus (M^{(i,1)}v^{(1)})_m \oplus (M^{(i,2)}v^{(2)})_m} \\
        \vdots \\
        &\xrightarrow{U_{i,k}} \ket{m, z \oplus \sum_{j\in[k]}(M^{(i,j)}v^{(j)})_m} = \ket{m, z\oplus ((Mv)^{(i)})_m} \\
    \end{align*}
    So that we have a unitary for the vector $(Mv)^{(i)} \in\F^d$. Using \Cref{concat_lemma}, we can get a unitary for $Mv \in \F^n$ that makes 1 query to each $(Mv)^{(i)}$. Each $(Mv)^{(i)}$ queries $M^{(i,j)}v^{(j)}$ once for each $j\in [k]$, and thus takes $O(k\log k)$ queries of $\sf ALG_4$. Consequently, $Mv$ makes $O(k^2\log k)$ queries to $\sf ALG_4$ and thus makes $ O((Q(n)+n^{3/2})\cdot \alpha^{-2}\cdot\log ^2(1/\alpha)\cdot \log\log(1/\alpha) )$ calls to $U_M$.
\end{proof}

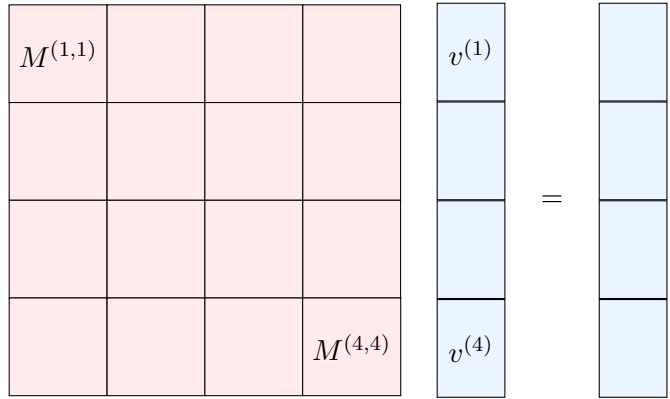
\begin{figure}[ht]
\centering
    \begin{tikzpicture}
        \definecolor{matrixcolor}{RGB}{255, 235, 235} 
        \definecolor{vectorcolor}{RGB}{235, 245, 255} 
    
        \matrix[matrix of math nodes, nodes={draw, anchor=center, minimum height=1.3cm, minimum width=1.3cm, fill=matrixcolor}, 
              column sep=-\pgflinewidth, row sep=-\pgflinewidth] (m) {
            M^{(1,1)} &  &  &  \\
             &  &  &  \\
             &  &  &  \\
             &  &  & M^{(4,4)} \\
        };
        \matrix[matrix of math nodes, nodes={draw, anchor=center, minimum height=1.3cm, minimum width=0.9cm, fill=vectorcolor}, 
                row sep=0pt, right=0.2cm of m] (v1) {
            v^{(1)} \\  \\  \\ v^{(4)} \\
        };
    
        \node[right=0.2cm of v1] (label) {$ =$};
    
        \matrix[matrix of math nodes, nodes={draw, minimum height=1.3cm, minimum width=0.9cm, fill=vectorcolor}, 
                row sep=0pt, right=0.2cm of label] (v) {
            \\ 
            \\ 
            \\ 
            \\
        };
    
    \end{tikzpicture}
    \caption{$M^{(i,j)}$ and $v^{(j)}$ for $k=4$.}
    \label{fig:3}
\end{figure}

\bibliographystyle{alpha}
\bibliography{references.bib}

\newcommand{\etalchar}[1]{$^{#1}$}
\begin{thebibliography}{AGG{\etalchar{+}}24}

\bibitem[AGG{\etalchar{+}}24]{Asadi_Golovnev_Gur_Shinkar_Subramanian_2022}
Vahid~R. Asadi, Alexander Golovnev, Tom Gur, Igor Shinkar, and Sathyawageeswar
  Subramanian.
\newblock Quantum worst-case to average-case reductions for all linear
  problems.
\newblock In David~P. Woodruff, editor, {\em Proceedings of the 2024 {ACM-SIAM}
  Symposium on Discrete Algorithms, {SODA} 2024, Alexandria, VA, USA, January
  7-10, 2024}, pages 2535--2567. {SIAM}, 2024.

\bibitem[AGGS22]{Asadi_Golovnev_Gur_Shinkar_2022}
Vahid~R. Asadi, Alexander Golovnev, Tom Gur, and Igor Shinkar.
\newblock Worst-case to average-case reductions via additive combinatorics.
\newblock In Stefano Leonardi and Anupam Gupta, editors, {\em {STOC} '22: 54th
  Annual {ACM} {SIGACT} Symposium on Theory of Computing, Rome, Italy, June 20
  - 24, 2022}, pages 1566--1574. {ACM}, 2022.

\bibitem[BBF03]{beauregard2003quantumarithmeticgaloisfields}
Stéphane Beauregard, Gilles Brassard, and José~M. Fernandez.
\newblock Quantum arithmetic on galois fields, 2003.

\bibitem[BLR93]{Blum_Luby_Rubinfeld_1993}
Manuel Blum, Michael Luby, and Ronitt Rubinfeld.
\newblock Self-testing/correcting with applications to numerical problems.
\newblock {\em Journal of Computer and System Sciences}, 47(3):549–595,
  December 1993.

\bibitem[BVMS05]{Bergholm_Vartiainen_Mottonen_Salomaa_2004}
Ville Bergholm, Juha~J. Vartiainen, Mikko Möttönen, and Martti~M. Salomaa.
\newblock Quantum circuits with uniformly controlled one-qubit gates.
\newblock {\em Physical Review A}, 71(5), May 2005.

\bibitem[GNW11]{Goldreich_Nisan_Wigderson_2011}
Oded Goldreich, Noam Nisan, and Avi Wigderson.
\newblock {\em On Yao’s XOR-Lemma}, volume 6650, page 273–301.
\newblock Springer Berlin Heidelberg, Berlin, Heidelberg, 2011.

\bibitem[HS23]{Hirahara_Shimizu_2023}
Shuichi Hirahara and Nobutaka Shimizu.
\newblock Hardness self-amplification: Simplified, optimized, and unified.
\newblock In {\em Proceedings of the 55th Annual ACM Symposium on Theory of
  Computing}, page 70–83, Orlando FL USA, June 2023. ACM.

\bibitem[Imp95]{Impagliazzo_1995}
Russell Impagliazzo.
\newblock Hard-core distributions for somewhat hard problems.
\newblock In {\em Proceedings of IEEE 36th Annual Foundations of Computer
  Science}, page 538–545, Milwaukee, WI, USA, 1995. IEEE Comput. Soc. Press.

\bibitem[IW97]{Impagliazzo_Wigderson_1997}
Russell Impagliazzo and Avi Wigderson.
\newblock \uppercase{P = BPP} if \uppercase{E} requires exponential circuits:
  derandomizing the \uppercase{XOR} lemma.
\newblock In {\em Proceedings of the twenty-ninth annual ACM symposium on
  Theory of computing - STOC ’97}, page 220–229, El Paso, Texas, United
  States, 1997. ACM Press.

\end{thebibliography}
\end{document}